\documentclass{ifacconf}
\usepackage{graphicx}      
\usepackage{natbib}        
\usepackage{amsmath,amssymb}
\usepackage{mathtools}
\usepackage{color}
\usepackage{subcaption}
\usepackage{tikz}
\usetikzlibrary{calc, positioning}

\usepackage{algorithm}
\usepackage{algpseudocode}

\makeatletter
\let\oldpf\pf
\let\endoldpf\endpf
\renewenvironment{pf}[1][\proofname]{\oldpf}{\hfill$\square$\endoldpf}
\makeatother

\DeclareMathOperator{\rank}{rank}

\DeclareMathOperator{\diag}{diag}

\DeclareMathOperator{\SO}{SO}

\begin{document}
\begin{frontmatter}


\title{Distributed Non-Uniform Scaling Control of Multi-Agent Formation with Dynamic Agent Joining}

\thanks[footnoteinfo]{This work was supported in part by the National Key Research and Development Program of China under Grant 2025YFA1018800 and in part by the National Natural Science Foundation of China under Grant 62573068 and Grant 62533002. {\it (Corresponding author: Gangshan Jing)}}

\author{Tao He, Gangshan Jing} 

\address{School of Automation, Chongqing University, Chongqing, China (e-mail: 20231301010@stu.cqu.edu.cn; jinggangshan@cqu.edu.cn).}

\begin{abstract}                

Non-uniform scaling control of formation enables multi-agent systems to adjust their shape by scaling with different ratios along different coordinate axes, offering enhanced flexibility in complex environments. However, like most existing formation maneuver strategies, it typically assumes a fixed set of agents, limiting its applicability in scenarios requiring dynamic team expansion. This paper introduces a distributed control framework that enables a formation to incorporate new agents during non-uniform scaling maneuvers in arbitrary dimensions while preserving the spectral properties of the graph Laplacian. Simulation examples validate the effectiveness of the theoretical results.
\end{abstract}

\begin{keyword}
Cooperative control, Formation control, Open multi-agent systems.
\end{keyword}

\end{frontmatter}

\section{Introduction}




Distributed formation maneuver control has become a foundational capability for multi-agent systems operating in complex environments. It enables a group of agents to function as a cohesive unit by coordinating a subset of them to guide the shape transformation of the entire formation using local sensing. This technology demonstrates significant potential in a variety of applications, ranging from search and exploration to cooperative manipulation (\cite{KAMEL2020128}). The realization of these advanced capabilities relies on effective control strategies, which primarily fall into two paradigms. Methods based on nonlinear constraints, such as distances (\cite{Yu2006,Trinh2024}) or angles (\cite{Jing2019,Buckley2021,Chen2021,Huang2026}), ensure rigid or uniform scaling geometry but often result in complex nonlinear controllers with challenges in stability analysis. Conversely, linear methods founded on the graph Laplacian encode the target formation in the matrix null space, yielding linear controllers and facilitating global convergence analysis.

Our work focuses on this linear Laplacian-based paradigm. While these methods have incorporated scaling maneuvers, the majority, such as those in (\cite{Lin2014, Li2018,Hector2021, Fang2024}), are confined to uniform scaling. This proves inefficient in constrained environments like narrow corridors, where uniform scaling forces unnecessary compression along all axes, whereas non-uniform scaling enables selective adaptation to spatial constraints. Affine formation (\cite{Lin2016, Zhao2018}) and the matrix-valued constraint framework proposed in our prior work (\cite{he2025}) both support non-uniform scaling of the formation geometry. However, the former typically relies on denser sensing graphs and a larger set of leaders, whereas the latter achieves this scaling flexibility with minimalist structural requirements.

Despite this advantage, the framework in \cite{he2025} is inherently restricted to a fixed set of agents, precluding its application to \textit{open multi-agent systems} (\cite{Franceschelli2021}) where agents may dynamically join or leave. This limitation stems from a fundamental challenge: unlike methods based on nonlinear constraints, Laplacian-based approaches do not naturally support network expansion, as any naive local modifications to the graph can alter the global spectral properties of the Laplacian matrix. This paper directly addresses this limitation by extending the framework in \cite{he2025} to support dynamic network expansion. The main contributions of this work are summarized as follows.

\begin{enumerate}
\item A theoretical framework for dynamic Laplacian construction that preserves the positive semi-definiteness and null space properties when new agents join the formation (see Theorem~\ref{thm:1}). 

\item A distributed agent joining protocol (see Algorithm~\ref{alg:agent_joining}) that realizes dynamic network expansion through local operations, eliminating the need for centralized computation as required in \cite{Lin2014,Zhao2018,he2025}.
\end{enumerate}



\section{Preliminaries and Problem Statement}


\subsection{Notation and Graph Theory}
The identity matrix is \( I_d \in \mathbb{R}^{d \times d} \), the all-ones vector is \( 1_d \in \mathbb{R}^d \), the zero tensor (scalar/vector/matrix) with context-appropriate dimensions is \( 0 \), the Kronecker product is \( \otimes \), and the symbols \( \succeq \) and \( \succ \) denote positive semi-definiteness and positive definiteness for matrices, respectively. The Special Orthogonal group is $\SO(d) = \{ R \in \mathbb{R}^{d \times d} \mid R^\top R = I_d, \det(R) = 1 \}$. For $x \in \mathbb{R}^d$, $\diag(x)$ is the $d \times d$ diagonal matrix with the components of $x$ on its main diagonal. Let $\|\cdot\|$ be the Euclidean norm.

Let a bidirectional graph $G = (V, E)$ model a multi-agent system, where the vertex set $V = \{1, \ldots, n\}$ represents the agents and the edge set $E \subseteq \{(i, j) : i, j \in V, i \neq j\}$ captures their sensing interactions. The graph is bidirectional, meaning that if $(i, j) \in E$, then $(j, i) \in E$. A directed edge $(j, i) \in E$ indicates that agent $i$ can sense agent $j$. The neighbor set of an agent $i$ is defined as $\mathcal{N}_i = \{j \in V : (j,i) \in E\}$.


The matrix-valued Laplacian \(L = [L_{ij}] \in \mathbb{R}^{dn \times dn}\) is a block matrix where each block \(L_{ij} \in \mathbb{R}^{d \times d}\) represents the \emph{matrix-valued edge weight} corresponding to the directed edge \((j, i) \in E\). The Laplacian satisfies the zero-row-sum property: \(\sum_{j=1}^n L_{ij} = 0\) for all \(i = 1, \dots, n\).

\begin{figure}[t]
	\centering
	\begin{tikzpicture}[scale=0.7]
	
	\begin{scope}

	\coordinate (A1) at (0,0);
	\coordinate (A2) at (1,0);
	\coordinate (A3) at (1,1);
	\coordinate (A4) at (0,1);
	
	\draw[dashed, thick] (A1) -- (A2) -- (A3) -- (A4) -- cycle;
	
	\foreach \point in {A1,A2,A3,A4} {
		\draw[fill=black] (\point) circle (0.06);
	}
	
	\draw[->, thick, olive, dashed, bend left=40] (1,1) to (4.1,1.5);
	
	\draw[fill=gray!30, line width=1pt] (3.5,0) -- (3.5,1.5) -- (3.7,1.5) -- (3.7,0) -- cycle;
	\draw[fill=gray!30, line width=1pt] (4.3,1) -- (4.3,2.5) -- (4.5,2.5) -- (4.5,1) -- cycle;
	\node[right] at (3,2.7) {Narrow Passage};

	\coordinate (C1) at (3.8,0.5);
	\coordinate (C2) at (4.2,0.5);
	\coordinate (C3) at (4.2,1.5);
	\coordinate (C4) at (3.8,1.5);
	
	\draw[dashed, thick, olive] (C1) -- (C2) -- (C3) -- (C4) -- cycle;
	
	\foreach \point in {C1,C2,C3,C4} {
		\draw[fill=olive] (\point) circle (0.06);
	}
	
	\draw[fill=yellow!30] (6.5+1,0.2) -- (8.5+1,0.2) -- (8.5+1,1.5) -- (6.5+1,1.5) -- cycle;
	\node[below] at (8,0) {Restricted Zone};
	
	\draw[->, thick, blue, dashed, bend right=15] (4.2,0.2) to (6+1,0.5);
	
	\coordinate (B1) at (6.1+1,0);
	\coordinate (B2) at (7.1+1,0);
	\coordinate (B3) at (7.1+1,1.7);
	\coordinate (B4) at (6.1+1,1.7);
	
	\draw[dashed, thick, blue] (B1) -- (B2) -- (B3) -- (B4) -- cycle;
	
	\foreach \point in {B1,B2,B3,B4} {
		\draw[fill=blue] (\point) circle (0.06);
	}
	
	\draw[->, thick] (0,0) -- (9.8,0) node[right] {$x$};
	\draw[->, thick] (0,0) -- (0,2.5) node[above] {$y$};
	
	\node[below] at (0.5,0) {Original};
	\node[below, olive] at (2.5,2.5) {(1) Scale $x$};
	\node[above, blue] at (5.5,0.3) {(2) Scale $y$};
	\end{scope}
	
	\node[above] at (4.5,-1.5) {(a) $R = I_2$};
		
	\begin{scope}[shift={(0,-5)}]
	
	\draw[->, thick] (0,0) -- (3,3) node[right] {$y'$};
	\draw[->, thick] (0,0) -- (1,-1) node[right] {$x'$};
	
	\coordinate (A1) at (0,0);
	\coordinate (A2) at (1,0);
	\coordinate (A3) at (1,1);
	\coordinate (A4) at (0,1);
	
	\draw[dashed, thick] (A1) -- (A2) -- (A3) -- (A4) -- cycle;
	
	\foreach \point in {A1,A2,A3,A4} {
		\draw[fill=black] (\point) circle (0.06);
	}
	
	\draw[fill=gray!30, line width=1pt] (2.5,1.0) -- (3.5,2.0) -- (3.7,1.8) -- (2.7,0.8) -- cycle;
	\draw[fill=gray!30, line width=1pt] (3.8,0.5) -- (4.8,1.5) -- (5.0,1.3) -- (4.0,0.3) -- cycle;

	\coordinate (D1) at (0.6+3.0,1.0); 
	\coordinate (D2) at (0.6+3.75,1.25);
	\coordinate (D3) at (0.6+4.0,2.0);
	\coordinate (D4) at (0.6+3.25,1.75);
	
	\draw[dashed, thick, orange] (D1) -- (D2) -- (D3) -- (D4) -- cycle;
	
	\foreach \point in {D1,D2,D3,D4} {
		\draw[fill=orange] (\point) circle (0.06);
	}
	
	\draw[->, thick, orange, dashed, bend right=20] (1,0.5) to (3.5,0.9);
	
	\draw[fill=gray!30, line width=1pt] (5.0+2,1.5) -- (6.0+2,0.5) -- (6.2+2,0.7) -- (5.2+2,1.7) -- cycle;
	\draw[fill=gray!30, line width=1pt] (5.3+3,2) -- (6.3+3,1) -- (6.5+3,1.2) -- (5.5+3,2.2) -- cycle;
	\node[below] at (6.1,0.5) {Narrow Passage};
	
	\coordinate (E1) at (2.6+5.5,1.5-0.5);   
	\coordinate (E2) at (2.6+6.25,1.25-0.5); 
	\coordinate (E3) at (2.6+6.0,2.0-0.5);   
	\coordinate (E4) at (2.6+5.25,2.25-0.5); 
	
	\draw[dashed, thick, purple] (E1) -- (E2) -- (E3) -- (E4) -- cycle;
	
	\foreach \point in {E1,E2,E3,E4} {
		\draw[fill=purple] (\point) circle (0.06);
	}
	
	\draw[->, thick, purple, dashed, bend left=30] (4.75,1.75) to (7.8,2);
	
	\node[above] at (0.5,1.2) {Original};
	\node[below, orange] at (2.5,0.3) {(3) Scale $x'$};
	\node[below, purple] at (6.5,3) {(4) Scale $y'$};
	\node[above] at (4.5,-2.5) {(b) $R = \begin{bmatrix} \cos 45^\circ & \sin 45^\circ \\ -\sin 45^\circ & \cos 45^\circ \end{bmatrix}$};
	\end{scope}
	
	\end{tikzpicture}
	\caption{Non-uniform scaling transformations for obstacle avoidance in 2D. The formation adaptively scales under different coordinate systems: (a) With $R = I_2$, scaling along original axes; (b) With rotation matrix $R$, scaling along rotated axes $x'$ and $y'$.}
	\label{fig:scaling_2d_rotated_corrected}
\end{figure}
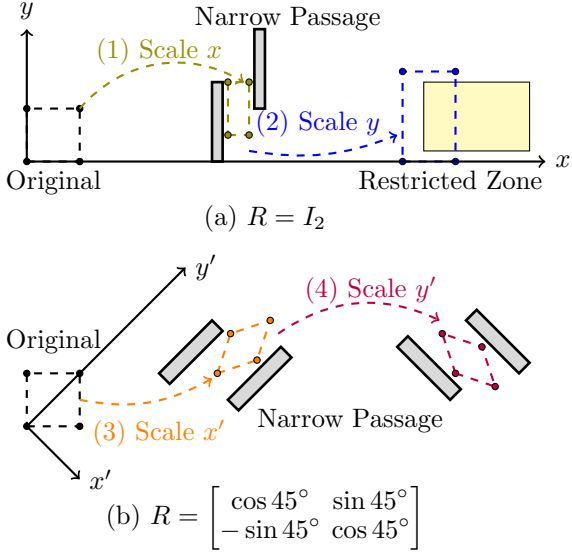

\subsection{Non-Uniform Scaling Formation Control}

Formation maneuver control under non-uniform scaling generalizes the concept of uniform scaling by allowing independent compression or stretching along different orthogonal directions specified by the orthogonal matrix $R \in SO(d)$, as illustrated in Fig.~\ref{fig:scaling_2d_rotated_corrected}. To formally describe such capability, consider a formation in $\mathbb{R}^d$ represented by $(G, p, R)$, where $G$ is the sensing graph, and the configuration $(p, R)$ consists of the stacked agent positions $p = [p_1^\top, \dots, p_n^\top]^\top \in \mathbb{R}^{dn}$ and the orthogonal matrix $R \in SO(d)$. Let $(\tilde{p}, R)$ denote a nominal configuration, where $\tilde{p} \in \mathbb{R}^{dn}$ is the nominal position vector that defines the reference shape of the formation. The desired shape space under non-uniform scaling is then defined as:
\begin{equation} \label{eq_dss}
    \varPi(\tilde p) \coloneqq \left\{ p = \left(I_n \otimes S(s, R)  \right) \tilde{p} + 1_n \otimes \tau \mid s, \tau \in \mathbb{R}^d \right\},
\end{equation}
where $S(s, R)=R \diag(s) R^{\top}$, $s$ and $\tau$ are time-varying scaling and translation maneuver parameters, respectively. 

A key challenge lies in the distributed implementation, as $\varPi$ depends on global parameters $(s, \tau)$ that are impractical to coordinate explicitly across all agents in a large-scale network. To address this, \cite{he2025} introduced a matrix-valued Laplacian $L \in \mathbb{R}^{dn \times dn}$ with $\ker(L) = \varPi$, and employed a leader-follower strategy to actively maneuver the formation within $\varPi$. Leaders ($V_l = \{1,\dots,m\}$) drive the formation, while followers ($V_f = \{m+1,\dots,n\}$) maintain the formation through constraints. Accordingly, by partitioning $p = [p_l^\top, p_f^\top]^\top$ and $
L = [\begin{smallmatrix} L_{ll} & L_{lf} \\ L_{fl} & L_{ff} \end{smallmatrix}]
$. Since $\ker(L) = \varPi$, any $p \in \varPi$ satisfies $L p = 0$. If the follower block $L_{ff}$ is non-singular, $p_f = -L_{ff}^{-1} L_{fl} p_l$. However, \cite{he2025} assumes a fixed sensing graph.

\subsection{Problem Statement}

This paper aims to develop a distributed formation control framework for open multi-agent systems using a Laplacian-based paradigm, which supports non-uniform scaling maneuvers while accommodating dynamic topological changes. The evolution of the system is modeled as a discrete sequence of nominal formations $F^k = (G^k, \tilde{p}^k, R)$ with associated Laplacians $L^k$, where the index $k \in \mathbb{N}$ labels the  formation after the $k$-th topology change.

The transition from $F^k$ to $F^{k+1}$, induced by arbitrary agent joining/leaving or edge addition/removal, yields the updated Laplacian
\begin{equation}\label{eq:L_update}
L^{k+1} = \mathcal{E}^k \bigl(L^k + \Delta^k\bigr) + \bar{\Delta}^k,
\end{equation}
where $\mathcal{E}^k: \mathbb{R}^{dn^k \times dn^k} \to \mathbb{R}^{dn^{k+1} \times dn^{k+1}}$ is the dimension-adjustment operator, to be specified later for agent joining/leaving, while $\Delta^k$ and $\bar{\Delta}^k$ encode local edge weight modifications applied before and after $\mathcal{E}^k$, respectively.

The convergence of the subsequent formation controller under the leader-follower strategy requires 
\begin{equation}\label{eq_spectral_properties}
L^k \succeq 0, \quad \ker(L^k) = \varPi(\tilde{p}^k), \quad L^k_{ff} \succ 0.
\end{equation}
To preserve this property after arbitrary topology changes, we exploit the closure of the positive semi-definite matrix cone under addition. Accordingly, we construct $\Delta^k$, and $\bar{\Delta}^k$ as symmetric positive semi-definite matrices.

Due to space constraints, this paper focuses on agent joining, formally defined as follows:

\begin{prob}[Agent Joining]
\label{pbm:agent_joining}
Given a nominal formation \( F^k = (G^k, \tilde{p}^k, R) \) in \( \mathbb{R}^d \) with Laplacian \( L^k \) satisfying \eqref{eq_spectral_properties}, design \( \mathcal{E}^k \), \( \Delta^k \), and \( \bar{\Delta}^k \) in \eqref{eq:L_update} for the joining of a new agent \( v^k \) with nominal position \( \tilde{p}_{v^k} \) using minimal neighbor connections, such that the updated Laplacian \( L^{k+1} \) satisfies \eqref{eq_spectral_properties} using only locally available information.
\end{prob}





To maintain notational conciseness, we omit the superscript $k$ in subsequent sections. For example, $L^+$ denotes $L^{k+1}$, $L$ denotes $L^k$, and $v$ denotes $v^k$.

\section{Distributed Framework for Dynamic Agent Joining} \label{sec:dynamic_agent_joining}
In this section, we first introduce the foundational matrix-valued constraints, then detail the positive semi-definite Laplacian construction, and finally integrate these components into a fully distributed algorithm.
\subsection{Matrix-Valued Constraints} \label{subsec:constraints}
To address problem \ref{pbm:agent_joining}, we propose a method for constructing a positive
semi-definite Laplacian based on matrix-valued constraints from \cite{he2025}. For a triplet of agents \((i,j,k)\), the constraint is defined as:
\begin{equation}
W_{jk} p_{ik}+W_{ki} p_{jk}=0,
\label{equ_vvc}
\end{equation}
where $p_{ik}=p_i-p_k$, $p_{jk}=p_j-p_k$, $W_{jk} = \diag\left(\tilde{p}_{jk,R}\right) R^{\top}$, $W_{ki} = \diag\left(\tilde{p}_{ki,R}\right) R^{\top}$, $\tilde{p}_{ki,R}=R^{\top} \tilde{p}_{ki}$, $\tilde{p}_{ki} = \tilde{p}_k - \tilde{p}_i$, $\tilde{p}_{jk,R}=R^{\top} \tilde{p}_{jk}$, $\tilde{p}_{jk} = \tilde{p}_j - \tilde{p}_k$.

\begin{lem}[\cite{he2025}] \label{lem:2}
	The constraint (\ref{equ_vvc}) is invariant to translation and non-uniform scaling transformation of $p_i, p_j, p_k$.
\end{lem}
Equation (\ref{equ_vvc}) can be rewritten as
\begin{equation}
    W_{kk}p_k=W_{jk} p_i + W_{ki} p_j, 
    \label{equ_vvc1}
\end{equation}
where $W_{kk} = W_{jk} + W_{ki}$. If $W_{kk}$ is non-singular, then the position of agent $k$ is uniquely determined by the positions of agents $i$ and $j$, i.e., $p_k = W_{kk}^{-1} (W_{jk} p_i + W_{ki} p_j)$. A key question is under what conditions $\rank(W_{kk}) = d$ holds.

\begin{lem} \label{lem:1}
   In~\eqref{equ_vvc1}, $p_k$ can be uniquely determined by $p_i$ and $p_j$ $\big($ i.e.,  $\rank(W_{kk})=d \big)$ if and only if $\prod_{l =1 }^d \tilde{p}^l_{ji, R}  \ne 0$.
\end{lem}
\begin{pf}
    Since $W_{kk} = \diag(\tilde{p}_{ji,R}) R^\top$ and $R$ is non-singular, it follows that $\rank(W_{kk}) = \rank(\diag(\tilde{p}_{ji,R}))$. The diagonal matrix $\diag(\tilde{p}_{ji,R})$ has rank $d$ if and only if $\tilde{p}^l_{ji,R} \neq 0$ for all $l = 1, \dots, d$.
\end{pf}

Lemma \ref{lem:1} motivates the following assumption.
\begin{assum} \label{ass:c}
    For each matrix-valued constraint~\eqref{equ_vvc} defined on triplet \((i,j,k)\), the condition \(\prod_{l=1}^d \tilde{p}^l_{ji,R} \neq 0\) holds.
\end{assum}

Our analysis relies on the following lemma.



\begin{lem}[\cite{horn2013matrix,zhang2005}]
\label{lem:inverse_schur}
For any symmetric block matrix of the form
\begin{equation}\nonumber
M = \begin{bmatrix}
A & B \\
B^\top & C
\end{bmatrix},
\end{equation}
where $C$ is invertible, let $S = A - B C^{-1} B^\top$ denote the Schur complement of $C$ in $M$. Then the following properties hold:

\begin{enumerate}
    \item The inverse of $M$ exists if and only if both $C$ and $S$ are invertible, in which case it is given by:
    \begin{equation} \nonumber
    M^{-1} = \begin{bmatrix}
    S^{-1} & ~~-S^{-1} B C^{-1} \\
    -C^{-1} B^\top S^{-1} & ~~C^{-1} + C^{-1} B^\top S^{-1} B C^{-1}
    \end{bmatrix}.
    \end{equation}
    \label{shur1}

    \item $M \succ 0$ if and only if $C \succ 0$ and $S \succ 0$.
    \label{shur2}

    \item If $C \succ 0$, then $M \succeq 0$ if and only if $S \succeq 0$.
    \label{shur3}
\end{enumerate}
\end{lem}

\subsection{Positive Semi-Definite Laplacian Construction} \label{subsec:laplacian}  
We now address Problem~\ref{pbm:agent_joining} by connecting \(v\) to two carefully selected agents $i$ and $j$ in $G$ such that the matrix-valued constraint among the triplet $(i,j,v)$ satisfies Assumption~\ref{ass:c}. The extended graph and nominal configuration are defined as:
\begin{equation}\label{eq:Gplus}
\begin{split}
  &G^+= (V\cup\{v\},E\cup\{(v,i),(i,v),(v,j),(j,v)\}), \\     
  &\tilde{p}^+=[\tilde{p}^\top,\tilde{p}_v^\top]^\top.
\end{split}
\end{equation}
The precise construction of the corresponding Laplacian matrix proceeds as follows. For the nominal formation \((G, \tilde{p}, R)\), we first reorder the agents in \(V\) to place \(i\) and \(j\) at the end. The corresponding Laplacian matrix \(L\) can be partitioned as
\begin{equation}
L = \begin{bmatrix}
L_{rr} & L_{rs} \\
L_{sr} & L_{ss}
\end{bmatrix},
\label{eq:L_block}
\end{equation}
where subscript \(r\) corresponds to agents in \(V \setminus \{i,j\}\) and \(s\) corresponds to the pair \(\{i,j\}\).

The geometric constraint for the triplet \((i,j,v)\) follows directly from \eqref{equ_vvc}, and can be written in matrix form as:
\begin{equation}
M_{ijv} p_{ijv} = 0,
\label{eq:iju_matrix}
\end{equation}
where \(M_{ijv} = \begin{bmatrix} W_{jv}, & W_{vi}, & -W_{vv} \end{bmatrix}\), \(W_{vv} = W_{jv} + W_{vi}\), \(p_{ijv} = [p_i^\top, p_j^\top, p_v^\top]^\top\).

To embed this constraint distributively into the Laplacian framework, we define the associated matrix block as:
\begin{equation}
L^{ijv} = M_{ijv}^\top D M_{ijv} \in \mathbb{R}^{3d \times 3d},
\end{equation}
where \(D \succ 0\) is a diagonal design matrix. The role of \(D\) is to provide a degree of freedom for weighting the constraint. This not only enables performance tuning but also establishes a mechanism for future network topology modifications, including agent leaving or edge addition/removal.

To incorporate agent \(v\), we set $\Delta = 0$ and apply the dimension-adjustment operator $\mathcal{E}$ to $L$:
\begin{equation}
\mathcal{E}(L) = L_{\text{pad}} = \begin{bmatrix} L_{rr} & L_{rs} & 0 \\ L_{sr} & L_{ss} & 0 \\ 0 & 0 & 0 \end{bmatrix} \in \mathbb{R}^{d|V^+| \times d|V^+|}.
\label{eq:L_pad}
\end{equation}


The term $\bar{\Delta}$ is constructed by embedding $L^{ijv}$ into a $d|V^+| \times d|V^+|$ zero matrix. In block form,
\begin{equation}\label{eq:Liju_pad}
\begin{split}
 \bar{\Delta} &= L^{ijv}_{\text{pad}} \\
&= \begin{bmatrix}
0 & 0 & 0 & 0 \\
0 & W_{jv}^\top D W_{jv} & W_{jv}^\top D W_{vi} & -W_{jv}^\top D W_{vv} \\
0 & W_{vi}^\top D W_{jv} & W_{vi}^\top D W_{vi} & -W_{vi}^\top D W_{vv} \\
0 & -W_{vv}^\top D W_{jv} & -W_{vv}^\top D W_{vi} & W_{vv}^\top D W_{vv}
\end{bmatrix}.   
\end{split}
\end{equation}

The Laplacian of \((G^+, \tilde{p}^+, R)\) is then given by:
\begin{equation}
L^+ = L_{\text{pad}} + L^{ijv}_{\text{pad}}.
\label{eq:LVA}
\end{equation}


\begin{thm} \label{thm:1}
Let $(G, \tilde{p}, R)$ be a nominal formation in $\mathbb{R}^d$ satisfying $L \succeq 0$, $\ker(L) = \varPi(\tilde{p})$, and $L_{ff} \succ 0$. For a new agent $v$ with $\tilde{p}_v$ and selected $i, j \in V$ such that triplet $(i,j,v)$ satisfies Assumption~\ref{ass:c}, the construction with $\Delta = 0$, $\mathcal{E}$ as in \eqref{eq:L_pad}, and $\bar{\Delta}$ as in \eqref{eq:Liju_pad} yields Laplacian $L^+$ from~\eqref{eq:LVA} satisfying $L^+ \succeq 0$, $\ker(L^+) = \varPi(\tilde{p}^+)$, and $L^+_{ff} \succ 0$.
\end{thm}

\begin{pf}
Since \(L \succeq 0\) and \(L^{ijv} = M_{ijv}^\top D M_{ijv} \succeq 0\) (as \(D \succ 0\)), it follows that \(L_{\text{pad}} \succeq 0\) and \(L^{ijv}_{\text{pad}} \succeq 0\). Thus, \(L^+ = L_{\text{pad}} + L^{ijv}_{\text{pad}} \succeq 0\).

Next, we prove that \(L^+_{ff} \succ 0\). There are three possible cases overall, determined by whether the selected agents \(i\) and \(j\) are leaders or followers, as listed below.

\textbf{Case 1: Both \(i, j \in V_l\).} Then \(v\) connects only to leaders, so \(L^+_{ff}\) decouples as
\begin{equation}
L^+_{ff} = \begin{bmatrix} L_{ff} & 0 \\ 0 & W_{vv}^\top D W_{vv} \end{bmatrix},
\end{equation}
where \(L_{ff} \succ 0\) is the original follower submatrix. By Assumption \ref{ass:c} and Lemma \ref{lem:1}, \(\rank(W_{vv}) = d\), so \(W_{vv}\) is invertible and \(W_{vv}^\top D W_{vv} \succ 0\) (since \(D \succ 0\)). Thus, \(L^+_{ff} \succ 0\).

\textbf{Case 2: One of \(i, j \in V_l\), the other in \(V_f\).} Without loss of generality, assume \(i \in V_l\), \(j \in V_f\). Partition the original \(L_{ff}\) as
\begin{equation}
L_{ff} = \begin{bmatrix} L_{P1} & L_{P2} \\ L_{P2}^\top & L_{jj} \end{bmatrix} \succ 0,
\end{equation}
where \(L_{P1}\) is the block corresponding to \(V_f \setminus  {j}\), \(L_{P2} \) is the coupling block between \(V_f \setminus  {j}\) and \(j\), and \(L_{jj} \in \mathbb{R}^{d \times d}\) is the block corresponding to follower \(j\). Since \(L_{ff} \succ 0\), by part (2) of Lemma \ref{lem:inverse_schur}, the Schur complement \(L_{P1} - L_{P2} L_{jj}^{-1} L_{P2}^\top \succ 0\) and \(L_{jj} \succ 0\).

After adding \(v\), we have
\begin{equation}
L^+_{ff} = \begin{bmatrix} L_{P1} & L_{P2} & 0 \\ L_{P2}^\top & L_{jj} + L^{ijv}_{jj} & L^{ijv}_{jv} \\ 0 & L^{ijv}_{vj} & L^{ijv}_{vv} \end{bmatrix},
\end{equation}
where \(L^{ijv}_{jj} = W_{vi}^\top D W_{vi}\), \(L^{ijv}_{jv} = -W_{vi}^\top D W_{vv}\), \(L^{ijv}_{vj} = -W_{vv}^\top D W_{vi}\), and \(L^{ijv}_{vv} = W_{vv}^\top D W_{vv}\). 

Since \(L^{ijv}_{vv}=W_{vv}^\top D W_{vv} \succ 0\) as in Case 1, and the Schur complement of \(L^{ijv}_{vv}\) is
\begin{equation}
    \begin{split}
        &L_{jj} + L^{ijv}_{jj} - L^{ijv}_{jv} (L^{ijv}_{vv})^{-1} L^{ijv}_{vj} = L_{jj} + W_{vi}^\top D W_{vi} \\
        &- W_{vi}^\top D W_{vv} (W_{vv}^\top D W_{vv})^{-1} W_{vv}^\top D W_{vi}=L_{jj} \succ 0.
    \end{split}
    \nonumber
\end{equation}
By part (2) of Lemma \ref{lem:inverse_schur}, we have
\begin{equation}
S = \begin{bmatrix} L_{jj} + L^{ijv}_{jj} & L^{ijv}_{jv} \\ L^{ijv}_{vj} & L^{ijv}_{vv} \end{bmatrix} \succ 0.
\end{equation}

Now examine the Schur complement of $S$ in $L^+_{ff}$:
\begin{equation}
L_{P1} - \begin{bmatrix} L_{P2} & 0 \end{bmatrix} S^{-1} \begin{bmatrix} L_{P2}^\top \\ 0 \end{bmatrix}.
\end{equation}
The (1,1)-block of $S^{-1}$ is obtained by applying Lemma \ref{lem:inverse_schur} (1) to the block matrix $S$, yielding
\begin{equation}
S^{-1}_{11} = \left(L_{jj} + L^{ijv}_{jj} - L^{ijv}_{jv}(L^{ijv}_{vv})^{-1}L^{ijv}_{vj}\right)^{-1} = L_{jj}^{-1}.
\end{equation}
Thus,
\begin{equation}
L_{P1} - L_{P2} S^{-1}_{11} L_{P2}^\top = L_{P1} - L_{P2} L_{jj}^{-1} L_{P2}^\top \succ 0.
\end{equation}
Since $S \succ 0$ and its Schur complement is positive definite, by part (2) of Lemma \ref{lem:inverse_schur}, $L^+_{ff} \succ 0$.

\textbf{Case 3: Both \(i, j \in V_f\).} Partition the original \(L_{ff}\) as
\begin{equation}
L_{ff} = \begin{bmatrix} L_{P1} & L_{P2} \\ L_{P2}^\top & L_{ss} \end{bmatrix} \succ 0,
\end{equation}
where \(L_{P1}\) is the block corresponding to \(V_f \setminus \{i,j\}\), \(L_{P2}\) is the coupling block between \(V_f \setminus \{i,j\}\) and the pair \(\{i,j\}\), and \(L_{ss} \in \mathbb{R}^{2d \times 2d}\) is the block corresponding to followers \(i\) and \(j\). Since \(L_{ff} \succ 0\), by part (2) of Lemma \ref{lem:inverse_schur}, the Schur complement \(L_{P1} - L_{P2} L_{ss}^{-1} L_{P2}^\top \succ 0\) and \(L_{ss} \succ 0\).

After adding \(v\), \(L^+_{ff}\) is
\begin{equation}
L^+_{ff} = \begin{bmatrix} L_{P1} & L_{P2} & 0 \\ L_{P2}^\top & L_{ss} + L^{ijv}_{ss} & L^{ijv}_{sv} \\ 0 & L^{ijv}_{vs} & L^{ijv}_{vv} \end{bmatrix},
\end{equation}
where 
\begin{equation}
    \begin{split}
        &L^{ijv}_{ss} = \begin{bmatrix}
                W_{jv}^\top D W_{jv} & W_{jv}^\top D W_{vi} \\
                W_{vi}^\top D W_{jv} & W_{vi}^\top D W_{vi}
                \end{bmatrix}, \,
                L^{ijv}_{sv} = \begin{bmatrix}
                -W_{jv}^\top D W_{vv} \\
                -W_{vi}^\top D W_{vv}
                \end{bmatrix}, \\
        &L^{ijv}_{vs} = -\begin{bmatrix}
                W_{vv}^\top D W_{jv} & W_{vv}^\top D W_{vi}
                \end{bmatrix}, \,
                L^{ijv}_{vv} = W_{vv}^\top D W_{vv}.
    \end{split}
    \nonumber
\end{equation}
Consider the bottom-right block \(S = \begin{bmatrix} L_{ss} + L^{ijv}_{ss} & L^{ijv}_{sv} \\ L^{ijv}_{vs} & L^{ijv}_{vv} \end{bmatrix}\). Since \(L^{ijv}_{vv} \succ 0\) and its Schur complement $L_{ss} + L^{ijv}_{ss} - L^{ijv}_{sv}(L^{ijv}_{vv})^{-1}L^{ijv}_{vs} = L_{ss} \succ 0$, by part (2) of Lemma \ref{lem:inverse_schur}, \(S \succ 0\). Now examine the Schur complement of \(S\) in \(L^+_{ff}\):
\begin{equation}
L_{P1} - \begin{bmatrix} L_{P2} & 0 \end{bmatrix} S^{-1} \begin{bmatrix} L_{P2}^\top \\ 0 \end{bmatrix}.
\end{equation}
The (1,1)-block of $S^{-1}$ is obtained by applying Lemma \ref{lem:inverse_schur} (1) to the block matrix $S$, yielding:
\begin{equation}
S^{-1}_{11} = \left(L_{ss} + L^{ijv}_{ss} - L^{ijv}_{sv}(L^{ijv}_{vv})^{-1}L^{ijv}_{vs}\right)^{-1} = L_{ss}^{-1}.
\end{equation}
Thus:
\begin{equation}
L_{P1} - L_{P2} S^{-1}_{11} L_{P2}^\top = L_{P1} - L_{P2} L_{ss}^{-1} L_{P2}^\top \succ 0.
\end{equation}
Since \(S \succ 0\) and its Schur complement is positive definite, by part (2) of Lemma \ref{lem:inverse_schur}, \(L^+_{ff} \succ 0\).

Having established \(L^+_{ff} \succ 0\) in all cases, we immediately obtain the rank condition:
\begin{equation}
\text{rank}(L^+) = \text{rank}(L) + d.
\end{equation}

We now prove $\ker(L^+) = \varPi(\tilde{p}^+)$. First, by construction and Lemma \ref{lem:2}, any $p^+ \in \varPi(\tilde{p}^+)$ satisfies:
\begin{equation}
L_{\text{pad}} p^+ = 0 \quad \text{and} \quad L^{ijv}_{\text{pad}} p^+ = 0,
\end{equation}
and thus $L^+ p^+ = 0$. This shows $\varPi(\tilde{p}^+) \subseteq \ker(L^+)$.

For the reverse inclusion, we compare dimensions. By the rank-nullity theorem and the established rank condition:
\begin{equation}
\dim(\ker(L^+)) = d(n+1) - \text{rank}(L^+) = dn - \text{rank}(L).
\end{equation}
Since $\ker(L) = \varPi(\tilde{p})$ for the nominal formation, we have $\dim(\ker(L)) = dn - \text{rank}(L)$. Therefore:
\begin{equation}
\dim(\ker(L^+)) = \dim(\varPi(\tilde{p}^+)).
\end{equation}
As $\varPi(\tilde{p}^+)$ is a subspace of $\ker(L^+)$ and their dimensions are equal, we conclude:
\begin{equation}
\ker(L^+) = \varPi(\tilde{p}^+).
\end{equation}
This completes the proof.
\end{pf}

\begin{algorithm}[htbp]
\caption{Distributed Agent Joining Protocol}
\label{alg:agent_joining}
\begin{algorithmic}[1]
\Require New agent $v$ with $\tilde{p}_v$; formation $(G, \tilde{p}, R)$ with $L \succeq 0$, $\ker(L) = \varPi(\tilde{p})$, $L_{ff} \succ 0$; each agent $i$ maintains $\{L_{ij} : j \in \mathcal{N}_i\}$
\Ensure Extended formation $(G^+, \tilde{p}^+, R)$ with $L^+ \succeq 0$, $\ker(L^+) = \varPi(\tilde{p}^+)$, $L^+_{ff} \succ 0$

\State Agent $v$ discovers candidates in sensing range and broadcasts join request
\State Candidates respond with nominal configurations; $v$ selects $i,j$ from candidates such that triplet $(i,j,v)$ satisfies Assumption~\ref{ass:c}

\State Agent $v$ establishes edge connections with $i$ and $j$, forming $(G^+, \tilde{p}^+, R)$ where $G^+ \leftarrow (V \cup \{v\}, E \cup \{(v,i),(i,v),(v,j),(j,v)\})$, $\tilde{p}^+ \leftarrow [\tilde{p}^\top, \tilde{p}_v^\top]^\top$

\State Each agent $k \in \{i, j, v\}$ computes and updates its local blocks of $L^+= [L^+_{ij}]$: 
\begin{equation}
\{L^+_{kl} \leftarrow L^{pad}_{kl} + \Delta L_{kl} \in \mathbb{R}^{d \times d} : l \in \mathcal{N}_k^+\}
\nonumber
\end{equation}
where $L^{pad} = \begin{bmatrix} L & 0 \\ 0 & 0 \end{bmatrix}$; the increments $\Delta L_{kl}$ are derived from \eqref{eq:Liju_pad} and $\mathcal{N}_i^+ \leftarrow \mathcal{N}_i \cup \{v\}$, $\mathcal{N}_j^+ \leftarrow \mathcal{N}_j \cup \{v\}$, $\mathcal{N}_v^+ \leftarrow \{i,j\}$

\State \textbf{return} $(G^+, \tilde{p}^+, R)$ with $L^+$
\end{algorithmic}
\end{algorithm}

\begin{rem}
Theorem \ref{thm:1} can be applied iteratively to add new agents, allowing the original formation to be extended to include more agents while preserving the desired properties of the Laplacian matrix. In contrast to methods relying on centralized optimization in (\cite{Mukherjee2020, Xiao2022}), our approach constructs the formation based on purely local, matrix-valued constraints \eqref{eq:iju_matrix}. Compared with the approach in \cite{Yang2019,Li2025}, which require each new agent to connect to $d+1$ existing agents with scalar edge weights and are limited to $\mathbb{R}^2$ and $\mathbb{R}^3$, our method connects each new agent to only two existing agents using matrix-valued edge weights, achieving a unified constructive procedure in arbitrary dimension $\mathbb{R}^d$.

\end{rem}

\subsection{Distributed Agent Joining Protocol}
Building upon the matrix-valued constraints (Subsection~\ref{subsec:constraints}) and the Laplacian construction method (Subsection~\ref{subsec:laplacian}), we now integrate these components into a executable distributed protocol. Algorithm~\ref{alg:agent_joining} outlines the step-by-step process for a new agent to dynamically join an existing formation. The protocol is fully distributed, requiring only local communication and computation; notably, the new agent needs no global knowledge of the network, agent indexing, or leader/follower roles.

\section{Simulation}
This section validates the theoretical results through two numerical experiments.

\subsection{Construction of Positive Semi-Definite Laplacian with Agent Joining} \label{sec:sim1}
This subsection validates the positive semi-definite Laplacian construction method described in Subsection~\ref{subsec:laplacian}. 
We begin with a nominal configuration in $\mathbb{R}^2$ with $R = I_2$ and the following nominal positions:
\begin{align*}
\tilde{p}_1 &= [-3, 3]^\top, \quad \tilde{p}_2 = [3, 2]^\top, \quad \tilde{p}_3 = [2, 0]^\top, \\
\tilde{p}_4 &= [1, -1]^\top, \quad \tilde{p}_5 = [0, -2]^\top, \quad \tilde{p}_6 = [-2, -3]^\top.
\end{align*}
The original Laplacian matrix $L$ is constructed to satisfy $L \succeq 0$, $\ker(L) = \varPi(\tilde{p})$, and $L_{ff} \succ 0$, with the corresponding edge weights shown in the left part of Fig.~\ref{fig:cycle}.

We add agent 7 at $\tilde{p}_7 = [-1, -2.1]^\top$ by connecting it to agents $5$ and $6$ to form the triplet $(5,6,7)$ as per Assumption~\ref{ass:c}. With $D = I_2$ and relative positions $\tilde{p}_{67,R} = [-1, -0.9]^\top$, $\tilde{p}_{75,R} = [-1, -0.1]^\top$, we compute the weight matrices $W_{67} = -\diag([1, 0.9]^{\top})$, $W_{75} = -\diag([1, 0.1]^{\top})$, $W_{77} = W_{67} + W_{75} = -\diag([2, 1]^{\top})$, and obtain the edge weights as shown in the middle part of Fig.~\ref{fig:cycle}:
\begin{align*}
L_{56} &= W_{67}^\top D W_{75} = \diag([1, 0.09]^{\top}), \\
L_{57} &= -W_{67}^\top D W_{77} = \diag([-2, -0.9]^{\top}), \\
L_{67} &= -W_{75}^\top D W_{77} = \diag([-2, -0.1]^{\top}).
\end{align*}

The edge weights in the right part of Fig.~\ref{fig:cycle} result from adding the corresponding edge weights in the left and middle parts. The Laplacian $L^+$ maintains all desired spectral properties, validating the Theorem~\ref{thm:1}.

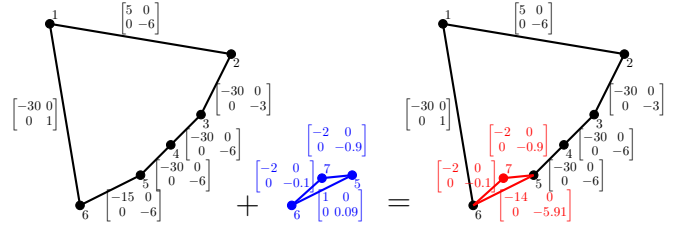
\begin{figure}[!t]
\centering

\begin{tikzpicture}[scale=0.4, every node/.style={scale=0.5}]
    \begin{scope}[local bounding box=graph1]
        \coordinate (v1) at (-3,3);
        \coordinate (v2) at (3,2);
        \coordinate (v3) at (2,0);
        \coordinate (v4) at (1,-1);
        \coordinate (v5) at (0,-2);
        \coordinate (v6) at (-2,-3);
        
        \foreach \point/\name/\pos in {v1/1/above, v2/2/below, v3/3/below, v4/4/below, v5/5/below, v6/6/below} {
            \filldraw (\point) circle (4pt) node[\pos, yshift=0pt, xshift=4pt] {\name};
        }
        
        \draw[thick] (v1) -- node[midway, above] {$\begin{bmatrix}
            5 & 0 \\
            0 & -6
        \end{bmatrix}$} (v2);
        \draw[thick] (v2) -- node[midway, right, xshift=-0.2cm, yshift=-0.4cm] {$\begin{bmatrix}
            -30 & 0 \\
            0 & -3
        \end{bmatrix}$} (v3);
        \draw[thick] (v3) -- node[midway, right, xshift=-0.2cm, yshift=-0.4cm] {$\begin{bmatrix}
            -30 & 0 \\
            0 & -6
        \end{bmatrix}$} (v4);
        \draw[thick] (v4) -- node[midway, right, xshift=-0.2cm, yshift=-0.4cm] {$\begin{bmatrix}
            -30 & 0 \\
            0 & -6
        \end{bmatrix}$} (v5);
        \draw[thick] (v5) -- node[midway, right, xshift=-0.3cm, yshift=-0.4cm] {$\begin{bmatrix}
            -15 & 0 \\
            0 & -6
        \end{bmatrix}$} (v6);
        \draw[thick] (v1) -- node[midway, left] {$\begin{bmatrix}
            -30 & 0 \\
            0 & 1
        \end{bmatrix}$} (v6);
    \end{scope}
    
    \node[right] at (3,-3) {\Huge$+$};

    \begin{scope}[local bounding box=triangle, shift={(7,0)}]
        \coordinate (v5) at (0,-2);
        \coordinate (v6) at (-2,-3);
        \coordinate (v7) at (-1,-2.1);
        
        \foreach \point/\name/\pos in {v5/5/below, v6/6/below, v7/7/above} {
            \filldraw[blue] (\point) circle (4pt) node[\pos, yshift=0pt, xshift=4pt] {\name};
        }
        
        \draw[thick, blue] (v5) -- node[midway, right, xshift=-0.3cm, yshift=-0.4cm] {$\begin{bmatrix}
            1 & 0 \\
            0 & 0.09
        \end{bmatrix}$} (v6);
        \draw[thick, blue] (v6) -- node[midway, left, xshift=0.4cm, yshift=0.4cm] {$\begin{bmatrix}
            -2 & 0 \\
            0 & -0.1
        \end{bmatrix}$} (v7);
        \draw[thick, blue] (v5) -- node[midway, above, xshift=0cm, yshift=0.4cm] {$\begin{bmatrix}
            -2 & 0 \\
            0 & -0.9
        \end{bmatrix}$} (v7);
    \end{scope}
    
    \node[right] at (8,-3) {\Huge$=$};
    
    \begin{scope}[local bounding box=result, shift={(13,0)}]
        \coordinate (v1) at (-3,3);
        \coordinate (v2) at (3,2);
        \coordinate (v3) at (2,0);
        \coordinate (v4) at (1,-1);
        \coordinate (v5) at (0,-2);
        \coordinate (v6) at (-2,-3);
        \coordinate (v7) at (-1,-2.1);
        
        \foreach \point/\name/\pos in {v1/1/above, v2/2/below, v3/3/below, v4/4/below, v5/5/below, v6/6/below} {
            \filldraw (\point) circle (4pt) node[\pos, yshift=0pt, xshift=4pt] {\name};
        }
        \filldraw[red] (v7) circle (4pt) node[above, yshift=0pt, xshift=4pt] {7};
                
        \draw[thick] (v1) -- node[midway, above] {$\begin{bmatrix}
            5 & 0 \\
            0 & -6
        \end{bmatrix}$} (v2);
        \draw[thick] (v2) -- node[midway, right, xshift=-0.2cm, yshift=-0.4cm] {$\begin{bmatrix}
            -30 & 0 \\
            0 & -3
        \end{bmatrix}$} (v3);
        \draw[thick] (v3) -- node[midway, right, xshift=-0.2cm, yshift=-0.4cm] {$\begin{bmatrix}
            -30 & 0 \\
            0 & -6
        \end{bmatrix}$} (v4);
        \draw[thick] (v4) -- node[midway, right, xshift=-0.2cm, yshift=-0.4cm] {$\begin{bmatrix}
            -30 & 0 \\
            0 & -6
        \end{bmatrix}$} (v5);
        \draw[thick, red] (v5) -- node[midway, right, xshift=-0.3cm, yshift=-0.4cm] {$\begin{bmatrix}
            -14 & 0 \\
            0 & -5.91
        \end{bmatrix}$} (v6);
        \draw[thick] (v1) -- node[midway, left] {$\begin{bmatrix}
            -30 & 0 \\
            0 & 1
        \end{bmatrix}$} (v6);
        \draw[thick, red] (v6) -- node[midway, left, xshift=0.4cm, yshift=0.4cm] {$\begin{bmatrix}
            -2 & 0 \\
            0 & -0.1
        \end{bmatrix}$} (v7);
        \draw[thick, red] (v5) -- node[midway, above, xshift=0cm, yshift=0.4cm] {$\begin{bmatrix}
            -2 & 0 \\
            0 & -0.9
        \end{bmatrix}$} (v7);
    \end{scope}
\end{tikzpicture}

\caption{agent joining. The original graph (left) combines with a triangle subgraph (middle) to form the extended graph (right).}
\label{fig:cycle}
\end{figure}

\begin{figure}[!tbp]
    \centering

    \begin{subfigure}[b]{0.45\textwidth} 
        \centering
        \includegraphics[width=\linewidth]{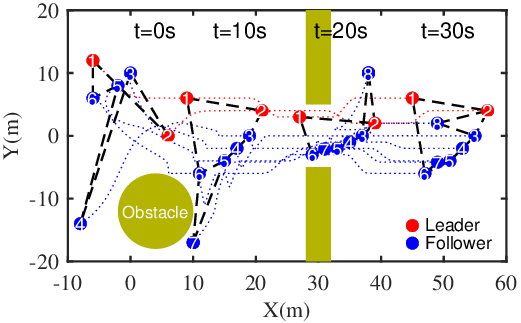}
        \caption{Agent trajectories.}
        \label{fig_te_va}
    \end{subfigure}
    \hfill
    \begin{subfigure}[b]{0.45\textwidth}
        \centering
        \includegraphics[width=\linewidth]{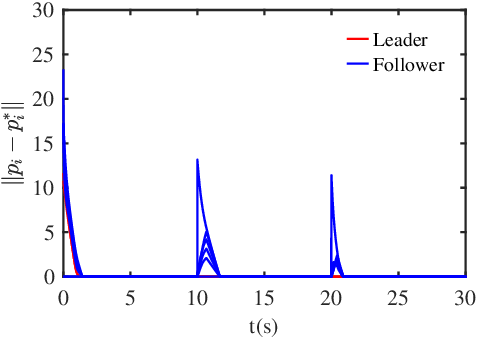}
        \caption{Tracking errors.}
        \label{fig_t_va}
    \end{subfigure}
    \caption{Formation maneuver control under agent joining (agent $7$ joins at $t=10$s, agent $8$ joins at $t=20$s).} 
    \label{fig_cycle_va}
\end{figure}

\subsection{Formation Maneuver Control with Dynamic Agent Joining}
This simulation couples Algorithm~\ref{alg:agent_joining} with the formation tracking control law from~\cite{Fang2022} to demonstrate non-uniform scaling maneuvers of a formation in a cluttered environment with dynamic agent joining.

For completeness, the control law is described as follows. Each leader $i \in V_l$ tracks a predefined trajectory via
\begin{equation}
u_i = -\alpha_1 \tanh(\alpha_2 (p_i - p^*_i)) + \dot{p}^*_i,
\label{equ_lu}
\end{equation}
where $p^*_i=S(s, R)\tilde{p}_i+\tau$, $\alpha_1,\alpha_2>0$ are control gains, and $\tanh(\cdot)$ denotes the hyperbolic tangent function. Each follower $i \in V_f$ maintains the formation using relative position measurements via
\begin{equation}
u_i = -\beta_1 \sum_{j \in N_i} L_{ij}(p_j-p_i) - \beta_2 \operatorname{sgn}\!\left(\sum_{j \in N_i} L_{ij}(p_j-p_i)\right),
\label{equ_pu}
\end{equation}
where $\beta_1,\beta_2>0$ are control gains, $\text{sgn}(\cdot)$ denotes the signum function, and $L_{ij}$ are the matrix-valued edge weights designed in previous sections.

The simulation timeline includes formation establishment (0--10 s), the integration of agent 7 at $10$ s, a non-uniform scaling maneuver (15--25 s), and the integration of agent 8 at $20$ s. As shown in Fig.~\ref{fig_cycle_va}(a), the initial six-agent formation achieves and maintains the desired shape during the establishment phase, using the nominal formation and Laplacian from Subsection~\ref{sec:sim1}. At $t = 10$ s, Agent 7 joins at nominal position $\tilde{p}_7 = [-1, -2.1]^\top$ by executing Algorithm~\ref{alg:agent_joining} with the triplet $(5,6,7)$, which triggers a one-time distributed update of the Laplacian matrix. From $t=15$ s to $t=25$ s, the formation performs a non-uniform scaling maneuver, compressing by 50\% along the y-axis to navigate a narrow passage. This specific maneuver is necessary since uniform scaling or rigid-body motion would result in inter-agent collisions or an inability to pass through the constriction. At $t = 20$ s, during the active scaling, agent 8 joins at $\tilde{p}_8 = [-1, 1]^\top$ by executing Algorithm~\ref{alg:agent_joining} with the triplet $(2,3,8)$ and $D = I_2$, testing the framework's ability to handle simultaneous shape deformation and network expansion.


Fig. \ref{fig_cycle_va} (b) shows that the tracking errors converge to zero, with only a minor transient disturbance during the agent joining event at $10$ s and $20$ s.


\section{Conclusion}
By integrating matrix-valued constraints with spectral graph theory, we have presented a distributed framework for the dynamic expansion of a formation in $\mathbb{R}^d$ that is capable of non-uniform scaling transformations while preserving the spectral properties of the graph Laplacian. Future work will extend the framework to more general topology modifications.



\bibliography{ifacconf}             

\begin{thebibliography}{22}
\providecommand{\natexlab}[1]{#1}
\providecommand{\url}[1]{\texttt{#1}}
\providecommand{\urlprefix}{URL }
\expandafter\ifx\csname urlstyle\endcsname\relax
  \providecommand{\doi}[1]{doi:\discretionary{}{}{}#1}\else
  \providecommand{\doi}{doi:\discretionary{}{}{}\begingroup
  \urlstyle{rm}\Url}\fi

\bibitem[{Buckley and Egerstedt(2021)}]{Buckley2021}
Buckley, I. and Egerstedt, M. (2021).
\newblock Infinitesimal shape-similarity for characterization and control of
  bearing-only multirobot formations.
\newblock \emph{IEEE Transactions on Robotics}, 37(6), 1921--1935.

\bibitem[{Chen et~al.(2021)Chen, Cao, and Li}]{Chen2021}
Chen, L., Cao, M., and Li, C. (2021).
\newblock Angle rigidity and its usage to stabilize multiagent formations in
  2-d.
\newblock \emph{IEEE Transactions on Automatic Control}, 66(8), 3667--3681.

\bibitem[{Fang et~al.(2022)Fang, Li, and Xie}]{Fang2022}
Fang, X., Li, X., and Xie, L. (2022).
\newblock Distributed formation maneuver control of multiagent systems over
  directed graphs.
\newblock \emph{IEEE Transactions on Cybernetics}, 52(8), 8201--8212.

\bibitem[{Fang and Xie(2024)}]{Fang2024}
Fang, X. and Xie, L. (2024).
\newblock Distributed formation maneuver control using complex laplacian.
\newblock \emph{IEEE Transactions on Automatic Control}, 69(3), 1850--1857.

\bibitem[{Franceschelli and Frasca(2021)}]{Franceschelli2021}
Franceschelli, M. and Frasca, P. (2021).
\newblock Stability of open multiagent systems and applications to dynamic
  consensus.
\newblock \emph{IEEE Transactions on Automatic Control}, 66(5), 2326--2331.

\bibitem[{{Garcia de Marina}(2021)}]{Hector2021}
{Garcia de Marina}, H. (2021).
\newblock Distributed formation maneuver control by manipulating the complex
  laplacian.
\newblock \emph{Automatica}, 132, 109813.

\bibitem[{He and Jing(2025)}]{he2025}
He, T. and Jing, G. (2025).
\newblock Distributed non-uniform scaling control of multi-agent formation via
  matrix-valued constraints.
\newblock \urlprefix\url{https://arxiv.org/abs/2508.02289}.

\bibitem[{Horn and Johnson(2013)}]{horn2013matrix}
Horn, R.A. and Johnson, C.R. (2013).
\newblock \emph{Matrix Analysis}.
\newblock Cambridge University Press.

\bibitem[{Huang and Jing(2026)}]{Huang2026}
Huang, J. and Jing, G. (2026).
\newblock Signed angle rigid graphs for network localization and formation
  control.
\newblock \emph{IEEE Transactions on Automatic Control}, 1--16.

\bibitem[{Jing et~al.(2019)Jing, Zhang, Lee, and Wang}]{Jing2019}
Jing, G., Zhang, G., Lee, H.W.J., and Wang, L. (2019).
\newblock Angle-based shape determination theory of planar graphs with
  application to formation stabilization.
\newblock \emph{Automatica}, 105, 117--129.

\bibitem[{Kamel et~al.(2020)Kamel, Yu, and Zhang}]{KAMEL2020128}
Kamel, M.A., Yu, X., and Zhang, Y. (2020).
\newblock Formation control and coordination of multiple unmanned ground
  vehicles in normal and faulty situations: A review.
\newblock \emph{Annual Reviews in Control}, 49, 128--144.

\bibitem[{Li et~al.(2026)Li, Chen, Wang, Li, and Shen}]{Li2025}
Li, H., Chen, H., Wang, X., Li, Z., and Shen, L. (2026).
\newblock Distributed framework construction for affine formation control.
\newblock \emph{IEEE Transactions on Automatic Control}, 71(1), 243--258.

\bibitem[{Li and Xie(2018)}]{Li2018}
Li, X. and Xie, L. (2018).
\newblock Dynamic formation control over directed networks using graphical
  laplacian approach.
\newblock \emph{IEEE Transactions on Automatic Control}, 63(11), 3761--3774.

\bibitem[{Lin et~al.(2014)Lin, Wang, Han, and Fu}]{Lin2014}
Lin, Z., Wang, L., Han, Z., and Fu, M. (2014).
\newblock Distributed formation control of multi-agent systems using complex
  laplacian.
\newblock \emph{IEEE Transactions on Automatic Control}, 59(7), 1765--1777.

\bibitem[{Lin et~al.(2016)Lin, Wang, Chen, Fu, and Han}]{Lin2016}
Lin, Z., Wang, L., Chen, Z., Fu, M., and Han, Z. (2016).
\newblock Necessary and sufficient graphical conditions for affine formation
  control.
\newblock \emph{IEEE Transactions on Automatic Control}, 61(10), 2877--2891.

\bibitem[{Mukherjee et~al.(2020)Mukherjee, Santilli, Gasparri, and
  Williams}]{Mukherjee2020}
Mukherjee, P., Santilli, M., Gasparri, A., and Williams, R.K. (2020).
\newblock Optimal topology selection for stable coordination of asymmetrically
  interacting multi-robot systems.
\newblock In \emph{2020 IEEE International Conference on Robotics and
  Automation (ICRA)}, 6668--6674.

\bibitem[{Vu et~al.(2024)Vu, Trinh, Tran, and Ahn}]{Trinh2024}
Vu, H.M., Trinh, M.H., Tran, Q.V., and Ahn, H.S. (2024).
\newblock Distance-based formation tracking of single- and double-integrator
  agents.
\newblock \emph{IEEE Transactions on Automatic Control}, 69(2), 1332--1339.

\bibitem[{Xiao et~al.(2022)Xiao, Yang, Zhao, and Fang}]{Xiao2022}
Xiao, F., Yang, Q., Zhao, X., and Fang, H. (2022).
\newblock A framework for optimized topology design and leader selection in
  affine formation control.
\newblock \emph{IEEE Robotics and Automation Letters}, 7(4), 8627--8634.

\bibitem[{Yang et~al.(2019)Yang, Cao, and Anderson}]{Yang2019}
Yang, Q., Cao, M., and Anderson, B.D.O. (2019).
\newblock Growing super stable tensegrity frameworks.
\newblock \emph{IEEE Transactions on Cybernetics}, 49(7), 2524--2535.

\bibitem[{Yu et~al.(2006)Yu, Fidan, Hendrickx, and Anderson}]{Yu2006}
Yu, C., Fidan, B., Hendrickx, J.M., and Anderson, B.D. (2006).
\newblock Merging multiple formations: A meta-formation prospective.
\newblock In \emph{Proceedings of the 45th IEEE Conference on Decision and
  Control}, 4657--4663.

\bibitem[{Zhang(2005)}]{zhang2005}
Zhang, F. (2005).
\newblock \emph{The {Schur} Complement and its Applications}.
\newblock Springer New York, NY, New York,USA.

\bibitem[{Zhao(2018)}]{Zhao2018}
Zhao, S. (2018).
\newblock Affine formation maneuver control of multiagent systems.
\newblock \emph{IEEE Transactions on Automatic Control}, 63(12), 4140--4155.

\end{thebibliography}

\end{document}